\documentclass[runningheads]{llncs}

\title{Analysis of Communication Channels Related to Physical Unclonable Functions}
\author{Georg Maringer\inst{1}\orcidID{0000-0002-2868-5131} \and Marvin Xhemrishi\inst{1} \and Sven Puchinger\inst{1} \and Kathrin Garb\inst{2} \and Hedongliang Liu\inst{1}\orcidID{0000-0001-7512-0654} \and Thomas Jerkovits\inst{1,3} \and Ludwig Kürzinger \and Matthias Hiller\inst{2}\orcidID{0000-0003-1238-1114} \and \\Antonia Wachter-Zeh\inst{1}\orcidID{0000-0002-5174-1947}
\thanks{G. Maringer's work was supported by the German Research Foundation (Deutsche Forschungsgemeinschaft, DFG) under Grant No. WA3907/4-1. The work of H.~Liu has been supported by a German Israeli Project Cooperation (DIP) grant under grant no.~PE2398/1-1 and KR3517/9-1. 
The work of S. Puchinger and A. Wachter-Zeh has been supported by the European Research
Council (ERC) under the European Union’s Horizon 2020 research and innovation programme (Grant Agreement No. 801434)
The work of K.~Garb and M.~Hiller has been supported by the DFG Priority Program SPP 2253 Nano Security (Project STAMPS).}}
\authorrunning{G. Maringer et al.}
\titlerunning{Analysis of Communication Channels related to PUFs}
\institute{Technical University of Munich, Germany \\
	\email{\{georg.maringer, marvin.xhemrishi, sven.puchinger, lia.liu,\\antonia.wachter-zeh\}@tum.de, ludwig.kuerzinger@mytum.de} \and
	Fraunhofer Institute for Applied and Integrated Security AISEC, Garching (near Munich), Germany \\
	\email{\{matthias.hiller, kathrin.garb\}@aisec.fraunhofer.de} \and
	Deutsches Luft- und Raumfahrtzentrum, Germany \\
	\email{\{thomas.jerkovits@dlr.de\}}
	}

\usepackage{amsmath, amssymb, amsbsy}
\usepackage{cleveref}
\usepackage{color}
\usepackage{bm}
\usepackage{pgfplots}
\usepackage{pifont}
\usepackage{float}
\usepackage{cite}
\usepackage{color}
\usepackage{amsfonts}
\usepackage{bbm}
\usepackage{mathtools}
\usepackage{setspace}
\usepackage{multicol}
\usetikzlibrary{positioning}
\usepackage[acronym]{glossaries}
\usepackage{xcolor}

\newcommand{\defeq}{\vcentcolon=}
\DeclareMathOperator{\supp}{supp}

\newcommand{\Exp}{\mathbb{E}}

\usepackage{tabularx}
\newcolumntype{P}[1]{>{\centering\arraybackslash}p{#1}}

\begin{document}

\maketitle

\begin{abstract}
Cryptographic algorithms rely on the secrecy of their corresponding keys.
On embedded systems with standard CMOS chips, where secure permanent memory such as flash is not available as a key storage, the secret key can be derived from Physical Unclonable Functions (PUFs) that make use of minuscule manufacturing variations of, for instance, SRAM cells.
Since PUFs are affected by environmental changes, the reliable reproduction of the PUF key requires error correction. For silicon PUFs with binary output, errors occur in the form of bitflips within the PUFs response. Modelling the channel as a Binary Symmetric Channel (BSC) with fixed crossover probability $p$ is only a first-order approximation of the real behavior of the PUF response. We propose a more realistic channel model, refered to as the Varying Binary Symmetric Channel (VBSC), which takes into account that the reliability of different PUF response bits may not be equal. We investigate its channel capacity for various scenarios which differ in the channel state information (CSI) present at encoder and decoder. We compare the capacity results for the VBSC for the different CSI cases with reference to the distribution of the bitflip probability according a work by Maes et al.
\end{abstract}

\newacronym{puf}{PUF}{Physical Unclonable Function}
\newacronym{vbsc}{VBSC}{Varying Binary Symmetric Channel}
\newacronym{bsc}{BSC}{Binary Symmetric Channel}

\section{Introduction}
\label{sec:intro}

Secure keys are the foundation for secure cryptographic operations. As a wide range of integrated circuits does not have access to secure key storage, \glspl{puf} evaluate physical properties of a circuit to derive a unique fingerprint and thus generate a device-specific cryptographic key. A \gls{puf} can be seen as the fingerprint of a physical object \cite{Mae12,HYKD14}, emanated from minuscule manufacturing variations that vary from object to object. This physical fingerprint can, hence, contribute to a hardware root of trust for security applications, such as authentication or identification of a device, and also the generation of a key for cryptographic primitives.

SRAM is widely available on embedded devices and a popular PUF primitive~\cite{GKST07}. Upon power-up, the state of SRAM cells is not deterministically defined, however they will mostly start up in the same state. Due to manufacturing variations, the state of a certain SRAM cell will either be ``0'' or ``1''. Multiple SRAM bits can therefore be grouped into a (randomly distributed) PUF-response. As characterized in \cite{MTV09a} and later refined in \cite{Mae13}, SRAM PUF cells vary significantly in their reliability. This reliability can be estimated with multiple measurements. In contrast, other PUFs such as the Ring-Oscillator PUF directly output analog or finely quantized discrete outputs, such that reliability information is directly available for each measurement. This also allows to estimate the reliability of a specific measurement as channel state information during decoding.

From these random SRAM bits, a key can be generated, which is exemplarily shown in Fig.~\ref{fig:keyGenOverview} for the code-offset fuzzy extractor \cite{Dodis2004}.
The PUF key is generated (once) during the manufacturing process of the device, which is referred to as \emph{enrollment}. After that, the PUF response is repeatedly read out whenever needed, such that the key can be reproduced.
However, since the \gls{puf} is subject to noise and also environmental effects such as temperature changes, certain SRAM bits can flip and impact the reliability of the key. Hence, an error correction algorithm was included into the key generation scheme.

Fig.~\ref{fig:keyGenOverview} shows that the key generation phase and the reproduction phase contain encoding and decoding steps. During the reproduction phase, the \gls{puf} will be read out again and the \gls{puf} response will deviate from the original one in the enrollment phase due to changes in environmental conditions. This corresponds to transmitting the PUF response over a noisy channel, which in case of SRAM cells can be described by a \gls{bsc}.

\begin{figure}[t!]
\begin{centering}
\resizebox{\textwidth}{!}{
\includegraphics{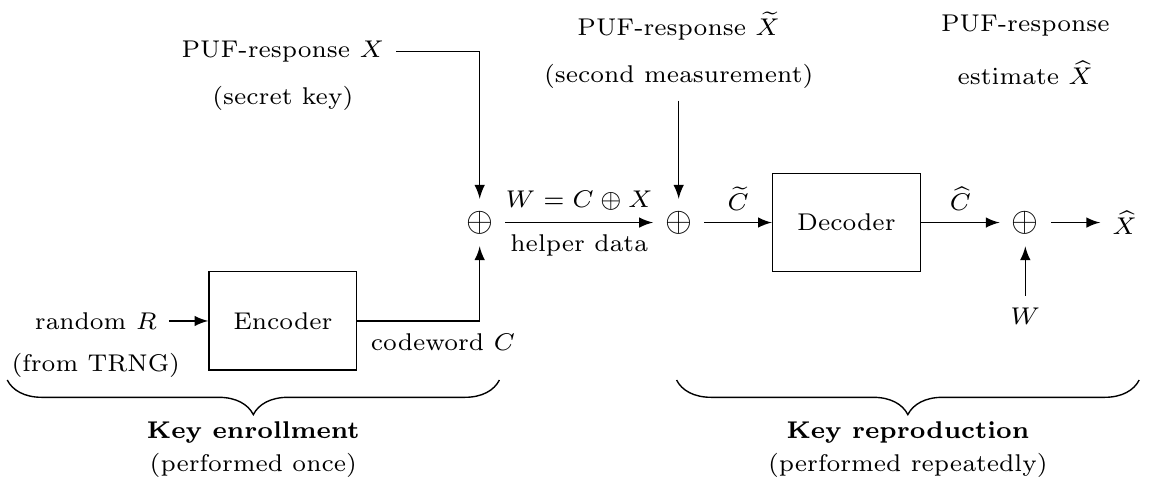}}
\caption{Simplified schematic of a key generation scheme based on \glspl{puf}.}
\label{fig:keyGenOverview}
\end{centering}
\end{figure}

The reliability of the power-up state of an SRAM cell heavily depends on the mismatch of between its transistors so that a \gls{bsc} model with constant error probability discards parts of the information. Experimental work compared in \cite{HKS20} has shown that fuzzy extractors benefit from reliability information, but still lack work on fundamental limits.

In this paper, we extend the \gls{bsc} model as used, e.g., in \cite{BIW+17}, by taking into account that the crossover probabilities may vary from cell to cell. This leads to the channel model discussed in Section~\ref{sec:channel_model} which we then analyze using information-theoretic methods in Section~\ref{sec:VBSC_analysis} distinguishing whether CSI is present at encoder and/or decoder. Furthermore within this section we show a capacity-achieving code construction using polar codes for CSI at encoder and decoder. Section~\ref{sec:example} shows the results of our analysis when the probability density function (pdf) of the crossover probabilities is specified according to \cite{MTV09a}. Section~\ref{sec:conclusion} concludes the paper.

\section{Notation}
We denote vectors by bold lowercase letters, e.g. $\bm{v}$, and its $i$-th component by $v_i$. We denote the subvector consisting of the elements from $i$ to $j$ by $v_{i}^j\defeq (v_i,v_{i+1},\ldots,v_j)$ where we omit the $i$ if it is equal to $1$. We denote random variables by uppercase letters, e.g. $X$ denotes a random variable and we denote its probability mass function by $P_X$ if the random variable is discrete and its probability density function by $f_X(x)$ in case the random variable $X$ is continuous. We denote the expected value of a random variable $X$ by $\Exp\big[X\big]$.
In this work logarithms are to the base $2$ unless otherwise stated. We denote the entropy of a discrete random variable by $H(X)$ and the mutual information between two random variables by $I(X;Y)$. We denote the binary entropy function by $H_2(p):= -p \log p -(1-p) \log(1-p)$.

\section{Channel Model}\label{sec:channel_model}
In Section~\ref{sec:intro}, it has been motivated why reading out a PUF response corresponds to data transmission over a noisy channel. More specifically the channel corresponds to a binary symmetric channel (BSC) that changes its crossover probability $p$ before each channel use according to a pdf $f_P(p)$.
In the following, it is assumed that $f_P(p)$ is known to the transmitter and to the receiver. The pdf has to be estimated by the manufacturer because it may be highly dependent on the manufacturing process.

\begin{figure}
	\begin{center}
			\scalebox{0.9}{
			\includegraphics{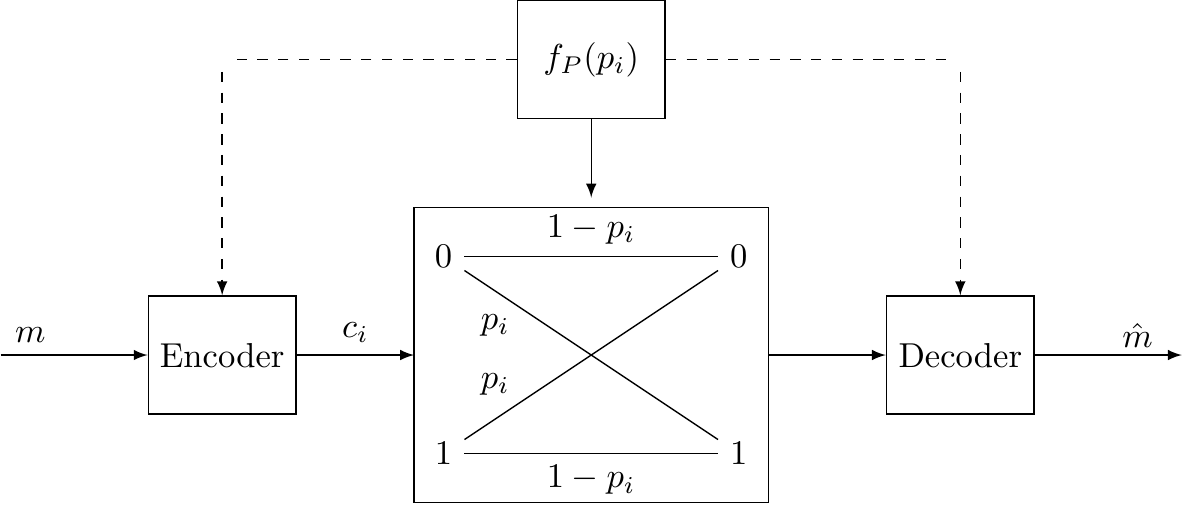}}
	\end{center}
	\caption{Varying Binary Symmetric Channel (VBSC) for the $i$-th channel use without knowledge of $p_i$ at encoder and decoder}
	\vspace{-1em}
	\label{fig:VBSC}
\end{figure}

For a block transmission with block length $n$, we denote the crossover probability for the $i$-th channel use by $p_i$. The $p_i$ can be given to the encoder and the decoder, potentially resulting in larger capacities for the varying binary symmetric channel (VBSC). There are several cases to be distinguished. The $p_i$ can be given to the encoder, the decoder, to both or to none of them. This is reflected in Fig.~\ref{fig:VBSC} by using dashed lines for the transmission of the $p_i$ to encoder and decoder. Furthermore, we distinguish causal and non-causal channel state information at the encoder.

\begin{definition}\label{def:causal_CSI}
	The channel state information (CSI) at the encoder is \emph{causal} if the encoder only has knowledge about $p^i \defeq p_1,\ldots,p_i$ to determine the $i$-th codeword symbol $c_i$, i.e., $c_i \defeq c_i(m,p^i)$.
\end{definition}
\begin{definition}\label{def:non-causal_CSI}
	The channel state information (CSI) at the encoder is \emph{non-causal} if the encoder has full knowledge about $p^n \defeq p_1,\ldots,p_n$ during the entire encoding process, i.e., $c_i \defeq c_i(m,p^n)$ for all $i$.
\end{definition}

Notice that it is not necessary to distinguish causal and non-causal CSI at the decoder. For the causal case the decoder can simply wait until the entire block $y^n$ and the complete CSI $p^n$ has been received before it starts the decoding procedure. Therefore, causal CSI is equivalent to non-causal CSI at the decoder.

On the contrary, causal and non-causal CSI at the encoder need to be distinguished in general.

\section{Information Theoretic Analysis of the channel}
\label{sec:VBSC_analysis}
In this section we aim to obtain the capacity for the VBSC. As introduced in Section~\ref{sec:channel_model} we distinguish several cases according to the availability of CSI at encoder and decoder. We present the capacity results for all cases except for CSI only non-causally at the encoder. The proofs within this section are frequently only shown for a continuous random variable $P$ describing the channel state. Practically, it may be more convenient to quantize the random variable $P$ since its concrete realization or its pdf (depending whether CSI is made available) need to be estimated. The proofs for a discrete random variable $P$ describing the channel state follow analogously.

We state the following lemma which is standard textbook knowledge (e.g. \cite{cover1999elements}) as a reminder for the reader.
\begin{lemma}[BSC Capacity, \cite{cover1999elements}]
	The capacity of the binary symmetric channel (BSC) with crossover probability $p$ is given by
	\begin{equation*}
	C_{BSC}(p) = 1-H_2(p)
	\end{equation*}
\end{lemma}
\subsection{No Channel State Information at Encoder and Decoder}\label{subsec:no_csi}
Drawing the $p_i$ randomly for each channel use can be interpreted as part of the channel's noise. We denote the input sequence of the channel by $X_1,\ldots,X_n$ and the corresponding outputs by $Y_1,\ldots,Y_n$. Since the VBSC is a discrete memoryless channel (DMC), it holds that
\begin{equation*}
	C = \max_{P_X} I(X;Y) \enspace .
\end{equation*}
\begin{theorem}\label{th:noCSI_encoder_decoder}
  Let $f_P(p)$ denote the pdf of the random variable $P$ from which the $p_i$ are drawn for a VBSC. Let $f_P(p)$ be known at encoder and decoder and let the realizations of the $p_i$ be unknown at encoder and decoder. Then, the capacity of the VBSC is

	\begin{equation*}
		C_{VBSC} = C_{BSC}\left(\Exp\big[P\big]\right) \enspace .
	\end{equation*}
\end{theorem}
\begin{proof}
	The statement can be proved by averaging over the pdf of $P$.
	It holds that
	\begin{align*}
		P_Y(0) &= \int_{\supp(f_P)} P_X(0) (1-p) f_P(p) \, dp + \int_{\supp(f_P)} (1-P_X(0)) p f_P(p) \, dp \nonumber\\
		&= P_X(0) + (1-2P_X(0)) \Exp\big[P\big]
	\end{align*}
	and
	\begin{align*}
		P_{Y|X}(0|0) &= \int_{\supp(f_P)} P_{Y|X,P} (0|0,p) f_P(p) \, dp = \int_{\supp(f_P)} (1-p) f_P(p) \, dp \nonumber\\
		&= 1-\Exp\big[P\big] \enspace .
	\end{align*}
	Furthermore, it holds that $P_{Y|X}(1|1)=P_{Y|X}(0|0)$ and hence $H(Y|X) = H_2\left(\Exp \big[ P \big]\right)$.
	Since the output alphabet is binary, we obtain
	\begin{equation*}
		I(X;Y) = H(Y) - H(Y|X) = H_2(P_Y(0)) - H_2(P_{Y|X}(0|0)) \enspace.
	\end{equation*}
	Notice that $H(Y|X)$ does not depend on the input distribution $P_X$. Therefore, we can determine the channel capacity by maximizing $H(Y)$ using the input distribution instead of maximizing $I(X;Y)$ for a generic DMC.
	Observe that choosing the input distribution to be uniform leads to a uniformly distributed output which maximizes the entropy $H(Y)$ and therefore achieves capacity.
	Thus,
	\begin{equation*}
		C_{VBSC} = 1-H_2\left(\Exp\big[P\big]\right) = C_{BSC}\left(\Exp\big[P\big]\right) \enspace .
	\end{equation*}
	\qed
\end{proof}

\subsection{Channel State Information at Encoder and Decoder}\label{subsec:state_encoder_decoder}
\begin{theorem}[\hspace{1pt}{\cite[Chapter 7.4]{el2011network}}]\label{th:elgamal_CSI_encoder_decoder}
  Let a discrete memoryless channel (DMC) with state space $\mathcal{P} \defeq \{1,\ldots,|\mathcal{P}|\}$ be given and let the state for each channel use be sampled i.i.d. from a probability mass function (pmf)
  $P_P(p)$. The capacity of this channel for CSI at encoder and decoder is given by
	\begin{equation*}
		C_{VBSC-ED} = \max_{P_{X|P}} I(X;Y|P) \enspace .
	\end{equation*}
\end{theorem}
This theorem is proved in the specified reference for a finite channel state space. The reason why it is not straightforward to generalize this result to infinite channel state spaces is that the achievability proof of the theorem uses rate splitting where the sequence $p_1,\ldots,p_n$ functions as a time sharing sequence. That means that the message $m$ is subdivided into $|\mathcal{P}|$ submessages $m_1,\ldots,m_{|\mathcal{P}|}$, which are then transmitted according to the CSI. In other words, for the transmission of $m_k$ the channel uses with state $k$ are utilized. The capacities conditioned on the respective channel states therefore determine the capacity of the combined channel with the respective distribution for the states $P_P(p)$ including CSI at encoder and decoder.
The converse result can be found in \cite{el2011network}.

Using Theorem~\ref{th:elgamal_CSI_encoder_decoder} we determine the capacity of the VBSC for finite channel state space $\mathcal{P}$.

\begin{theorem}\label{th:capacity_encoder_decoder}
	For the capacity of the VBSC with CSI at encoder and decoder and finite channel state space $\mathcal{P}$ it holds that
	\begin{equation*}
		C_{VBSC-ED} = \Exp \big [ C_{BSC}(P) \big ] \enspace .
	\end{equation*}
\end{theorem}
\begin{proof}
	Using Theorem~\ref{th:elgamal_CSI_encoder_decoder} we have that
	\begin{equation*}
		C = \max_{P_{X|P}} I(X;Y|P) = \max_{P_{X|P}} (H(Y|P) - H(Y|X,P)) \enspace .
	\end{equation*}
	By observing that
	\begin{equation*}
		P_{Y|X,P}(0|0,p) = P_{Y|X,P}(1|1,p) = 1-p
	\end{equation*}
	we recognize that $H(Y|X,P=p) = H_2(p)$ is independent of $P_{X|P}$. Thus, it holds that the maximization of the mutual information boils down to a maximization of a single entropy.
	It holds that
	\begin{equation*}
		P_{Y|P}(0|p) = P_{X|P}(0|p) (1-p) + (1-P_{X|P}(0|p)) p \enspace .
	\end{equation*}
	For $P_{X|P}(0|p)=1/2$ we obtain that $P_{Y|P}(0|p)=1/2$ irrespective of $p$. Therefore, the uniform distribution is capacity-achieving for the VBSC with CSI at encoder and decoder and the capacity of the channel is
	\begin{equation*}
		C_{VBSC-ED} = \max_{P_{X|P}} I(X;Y|P) = 1-\Exp\big [H_2(P) \big] = \Exp \big [ C_{BSC}\big ]
	\end{equation*}
	\qed
\end{proof}
Now we show that for many relevant continuous pdfs $f_P(p)$ describing the channel state we can approach the rate
\begin{equation*}
\tilde{R} \defeq \max_{P_{X|P}} I(X;Y|P) = 1-\int_{\supp(f_P)} H_2(p)f_P(p) dp = \Exp \big [ C_{BSC}(P) \big]
\end{equation*}
arbitrarily close.
\begin{corollary}\label{cor:encoder_decoder}
	Let $f_P(p)$ denote the pdf of a random variable $P$ from which the crossover probability of the VBSC is drawn. Furthermore, let $f_P(p)$ be differentiable with a continuous derivative in $(0,1)$. If CSI is known at encoder and decoder the capacity of the channel is arbitrarily close to $\tilde{R}$.
\end{corollary}
\begin{proof}
	We fix some $\varepsilon>0$ and show that the rate $\tilde{R}-\varepsilon$ is achievable.
	To do so first fix $p^*$ such that
	\begin{equation*}
	\int_0^{p^*} f_P(p) dp < \frac{\varepsilon}{3} \enspace .
	\end{equation*}
	Furthermore, we fix $p^{**}$ such that
	\begin{equation*}
		\int_{p^{**}}^1 f_P(p) dp < \frac{\varepsilon}{3} \enspace .
	\end{equation*}
	Next we compute the Riemann-lower sum for the interval $[p^*,p^{**}]$ with $n$ equidistant subintervals $S^{(n)}_{p^*,p^{**}}(C_{BSC}(p)f_P(p))$ of the function $C_{BSC}(p)f_P(p)$. We can bound the difference between the integral and the Riemann-lower sum by
	\begin{equation*}
	\left | \int_{p^*}^{p^{**}} C_{BSC}(p) f_P(p) dp - S_{[p^*,p^{**}]}^{(n)}(C_{BSC}(p)f_P(p)) \right | \leq \frac{A (p^{**}-p^*)^2}{n}
	\end{equation*}
	where $A \defeq \max_{[p^*,p^{**}]} \left|\frac{d}{dp}(f_P(p)C_{BSC}(p))\right |$.
	
	The maximization is well defined because the interval $[p^*,p^{**}]$ is bounded and closed and the derivative of $C_{BSC}(p)f_P(p)$ is continuous. By choosing $n$ large enough we can upper bound the absolute value of the difference of Riemann-lower sum and the integral by $\varepsilon/3$. Since the capacity of the BSC is upper bounded by $1$ and lower bounded by $0$, the difference between Riemann-lower sum and integral within the intervals $[0,p^*]$ and the interval $[p^{**},1]$ are also bounded by $\varepsilon/3$. Similar arguments can be made for the Riemann-upper sum to get an upper bound and therefore we approach $\tilde{R}$ arbitrarily close.
	\qed
\end{proof}

Theorem 1 in \cite{MTV09a} states that publication of the CSI does on average not leak any information about the expected response of the SRAM cell. Therefore, the requirement of having CSI at the decoder does not compromise the security of the SRAM-PUF and the channel model with CSI at the encoder and decoder is relevant in practice.

\begin{figure}[h]
	\centering
	\scalebox{0.8}{
		\includegraphics{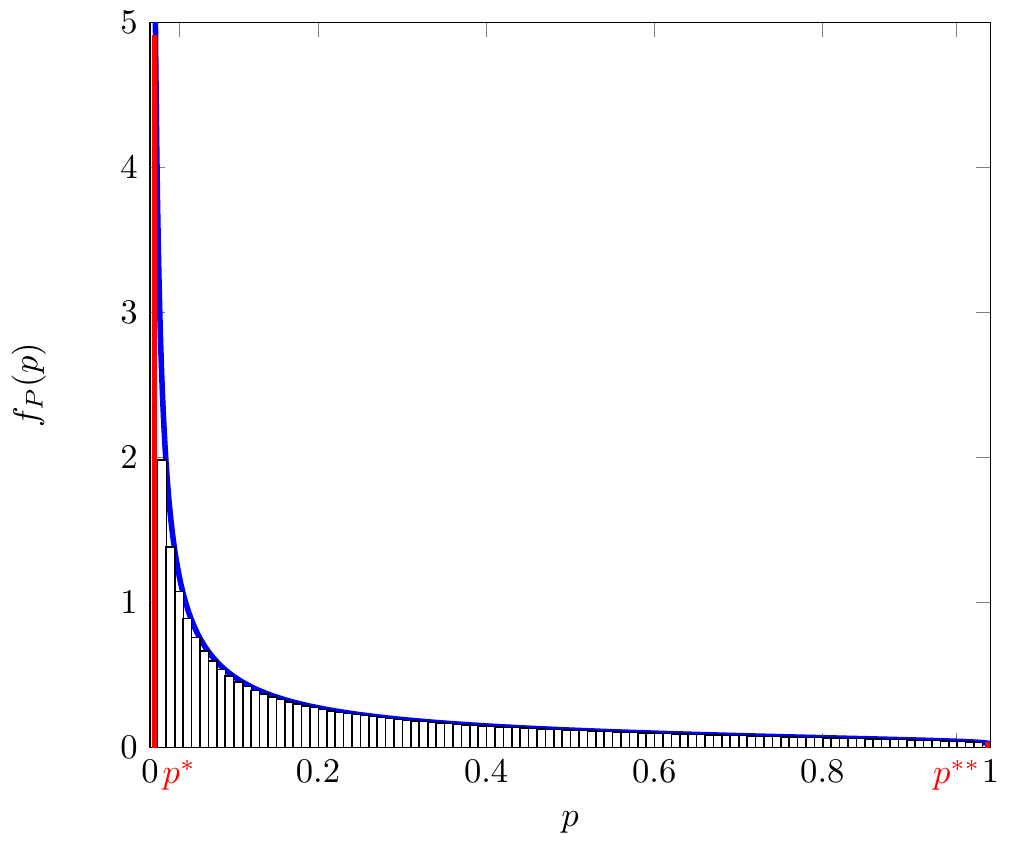}}
	\vspace{-1em}
	\caption{The probability density function (pdf) of the error for the SRAM cell. The pdf is computed as shown in~\cite{Mae13} with parameters $\lambda_1 = 0.1213$ and $\lambda_2 = 0.021$.}
	\label{fig:pdf}
	\vspace{-3em}
\end{figure}

\subsubsection{Achievability of Capacity using Polar Codes}

The proof of Corollary~\ref{cor:encoder_decoder} also shows that by constructing a finite number of polar codes, it is possible to achieve the capacity of the VBSC with CSI at encoder and decoder. Basically for each of the intervals in Fig.~\ref{fig:pdf}, it is possible to construct a polar code for a BSC with crossover probability equal to the value of the minimal $C_{BSC}(p)$ (largest $p$ if the interval is in $[0,1/2]$ and smallest $p$ if the interval is in $[1/2,1]$). Polar codes are capacity achieving for BSCs \cite{arikan2009channel}. Since every interval occurs with strictly positive probability, on average each of them occurs linearly in the blocklength $n$. Using the respective polar code for each interval achieves the capacity of the channel arbitrarily close.

\subsection{Channel State Information only at the Decoder}
As already mentioned in Section~\ref{sec:channel_model}, it makes no difference whether channel state information is given to the decoder in a causal or non-causal manner.
\begin{theorem}\label{th:csi_decoder}
	The capacity of the VBSC with channel state information at the decoder is equal to
	\begin{equation*}
	C_{VBSC-D} = \Exp \left[C_{BSC} (P) \right]
	\end{equation*}
\end{theorem}
\begin{proof}
	We can interpret the channel state information as an additional output of the channel \cite{el2011network}, i.e. we define the output by $Y' = (Y,P)$. This channel is still a DMC and hence its capacity is defined by
	\begin{equation*}
	C = \max_{P_X} I(X;Y') = \max_{P_X} I(X;Y,P) \enspace .
	\end{equation*}
	The remaining parts of the proof are similar to the proof of Theorem~\ref{th:capacity_encoder_decoder} with the exception that rate splitting is not necessary and therefore the proof of the statement when the channel state $P$ is drawn from a continuous random variable is the same as for the finite alphabet case.
	\qed
\end{proof}
Notice that CSI at the encoder (causally or non-causally) has no effect on the capacity of the VBSC as long as CSI is given to the decoder. This effect occurs because the input distribution has no effect on $H(Y|X,P)$ and the uniform distribution maximizes $H(Y|P)$ independent of $P$. Therefore, CSI is of no use in this case when it comes to maximizing the mutual information $I(X;Y|P)$ and consequently capacity is not increased by additional CSI at the encoder.

\subsection{Causal Channel State Information at the Encoder}\label{subsec:csi_causal_enc}
As discussed in Subsection~\ref{subsec:state_encoder_decoder}, publishing CSI does not compromise the security of the PUF. However, CSI only at the encoder may still be relevant because including the CSI into the helper data of the PUF increases the channel capacity at the expense of also increasing the helper data size.
In this subsection, we consider causal channel state information at the encoder.

In this subsection we will show the capacity of the VBSC with causal CSI at the encoder. Furthermore, we will show that the capacity is higher (except for some special cases) than for the case without CSI at encoder and decoder. This statement is non-trivial due to the symmetry of the channel.

\begin{theorem}\label{th:VBSC_cap_encoder}
	Let $f_P(p)$ denote the pdf from which the crossover probability is sampled i.i.d. for each channel use. Then the capacity of the VBSC with causal CSI is equal to
	\begin{equation}
	C_{VBSC-E} = 1-H_2 \left(\int_0^{0.5} f_P(p) (1-p) \, dp + \int_{0.5}^1 f_P(p)p \, dp \right) \enspace .
	\end{equation}
\end{theorem}

In order to prove this theorem, we make use of a generic result on the capacity of channels with causal CSI at the encoder which is stated in \cite{el2011network}.
\begin{theorem}\label{th:CSI_encoder_ElGamal}
	The capacity of a DMC with CSI causally available at the encoder can be computed by
	\begin{equation}
	C_{CSI-E} = \max_{P_U, x(U,P)} I(U;Y) \enspace ,
	\end{equation}
	where for the alphabet size $|\mathcal{U}|$ of the auxiliary variable $U$, it holds that $|\mathcal{U}| \leq \min\{(|\mathcal{X}|-1)|\mathcal{P}|+1, |\mathcal{Y}|\}$.
\end{theorem}

In the theorem above, the auxiliary variable $U$ can be interpreted as the output of an encoder which encodes an arbitrary message $m$ independently of the channel state. $U$ functions as the input to a mapping device that maps the input according to the channel state onto the channel's input alphabet. This resulting $X$ is then transmitted over the channel resulting in the channel output $Y$. The transmission of the $i$-th symbol within the block is illustrated in Fig.~\ref{fig:VBSC_causal}.
\begin{figure}[t]
	\begin{center}
		\resizebox{1.1\textwidth}{!}{
			\includegraphics{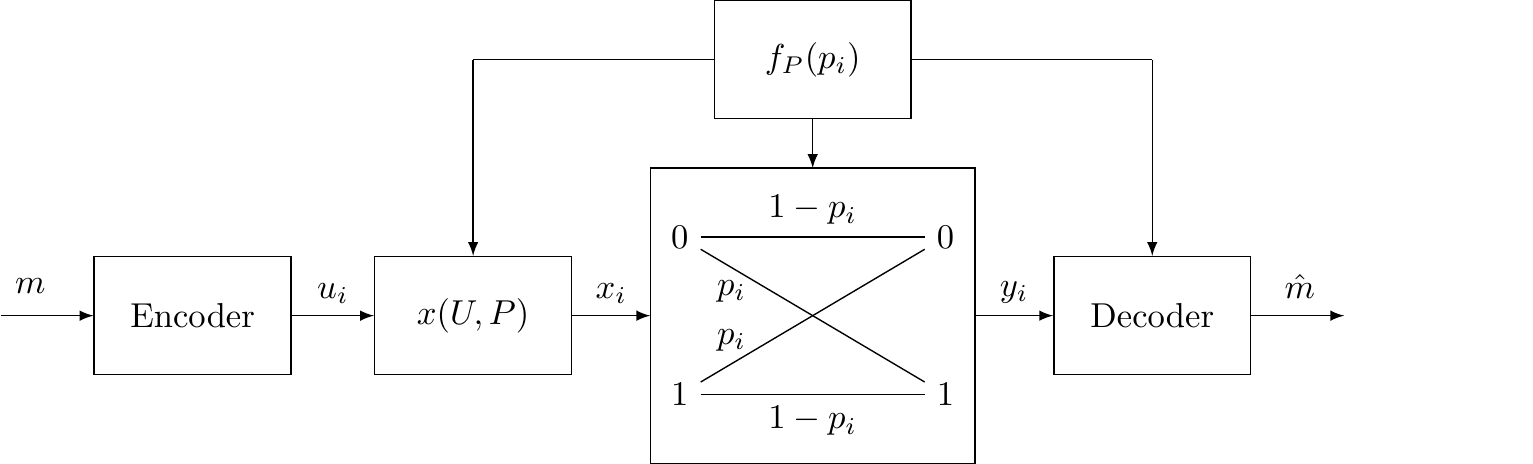}}
	\end{center}
	\caption{Analysis of the VBSC with causal CSI at the encoder according to Theorem~\ref{th:CSI_encoder_ElGamal}}
	\vspace{-1em}
	\label{fig:VBSC_causal}
\end{figure}
According to Theorem~\ref{th:CSI_encoder_ElGamal} for the VBSC we have that $|\mathcal{U}| \leq 2$ because the output alphabet of the channel is binary.

The proof of Theorem~\ref{th:VBSC_cap_encoder} is based on the following two lemmas.
\begin{lemma}\label{lem:cond_entropy}
	The mapper
	\begin{equation}\label{eq:mapper}
		x(u,p) = \begin{cases}
		u \quad \text{for } p \leq 0.5\\
		\overline{u} \quad \text{else}
		\end{cases}
	\end{equation}
	minimizes the conditional entropy $H(Y|U)$. Furthermore, $H(Y|U)$ is independent of $P_U$ and it holds
	\begin{equation}\label{eq:cond_entropy}
		H(Y|U) = H_2\left( \int_0^1 f_P(p) \max\{p,1-p\} \, dp\right) \enspace .
	\end{equation}
\end{lemma}
\begin{proof}
	By using the specified mapping device we observe that
	\begin{equation}\label{eq:P_YU}
		P_{Y|U}(0|0) = P_{Y|U}(1|1) = \int_0^1 f_P(p) \max\{p,1-p\}\, dp \enspace .
	\end{equation}
	Hence, we have
	\begin{equation}
		H(Y|U) = H_2\left( \int_0^1 f_P(p) \max\{p,1-p\}\, dp\right)
	\end{equation}
	independent of $P_U$.
	\qed
\end{proof}
\begin{lemma}\label{lem:input_distribution}
	If the mapper is chosen as specified in Lemma~\ref{lem:cond_entropy}, choosing $P_U$ to be the uniform distribution leads to a uniform output distribution $P_Y$.
\end{lemma}
\begin{proof}
	We observe that
	\begin{align}
		P_Y(0) &= P_U(0) \int_0^{0.5} f_P(p) (1-p) \, dp + P_U(0) \int_{0.5}^1 f_P(p) p \, dp \nonumber \\
		&\quad + (1-P_U(0)) \int_0^{0.5} f_P(p) p \, dp + (1-P_U(0)) \int_{0.5}^1 f_P(p) (1-p) \, dp \enspace .
	\end{align}
	By choosing $P_U(0) = 1/2$ we obtain $P_Y(0)=1/2$, achieving $H(Y)=1$.
	\qed
\end{proof}
Finally, we are ready to determine the capacity of the VBSC with causal CSI at the encoder.
\begin{proof}[Theorem~\ref{th:VBSC_cap_encoder}]
	Expanding the mutual information in Theorem~\ref{th:CSI_encoder_ElGamal} we have
	\begin{equation*}
		C_{VBSC-E} = \max_{P_U, x(U,P)} I(U;Y) = \max_{P_U, x(U,P)} H(Y) - H(Y|U) \enspace .
	\end{equation*}
	By Lemma~\ref{lem:cond_entropy} we have that the conditional entropy $H(Y|U)$ is minimized for the mapper specified in~\eqref{eq:mapper}. Furthermore, according to Lemma~\ref{lem:input_distribution}, a uniform $P_U$ leads to a uniform $P_Y$ and consequently $H(Y)=1$. Therefore, the choice we made for $P_U$ and $x(U,P)$ maximizes and $I(U;Y)$ and the channel capacity is obtained.
	Combining the previous results, we get
	\begin{equation*}
		C_{VBSC-E} = 1-H_2 \left(\int_0^{0.5} f_P(p) (1-p) \, dp + \int_{0.5}^1 f_P(p)p \, dp \right)
	\end{equation*}
	completing the proof.
	\qed
\end{proof}

\subsubsection{Comparison to the capacity without CSI:}
We show that $C_{VBSC} \leq C_{VBSC-E}$ with equality if and only if $f_P(p)=0, \; \forall p \in (1/2,1]$ almost everywhere.
The capacity of the VBSC without CSI at encoder and decoder has been specified in Theorem~\ref{th:noCSI_encoder_decoder} by
\begin{equation*}
C_{VBSC} = C_{BSC}(\Exp\big[ P \big ]) = 1-H_2(\Exp\big [ P \big ]) = 1-H_2\left( \int_0^1 f_P(p) p \, dp \right)\enspace .
\end{equation*}
Recall that $H_2(p)$ is concave and symmetric around its maximum which is achieved for $p=1/2$. Observe that
\begin{equation*}
\int_0^{0.5} f_P(p)(1-p) \, dp + \int_{0.5}^1 f_P(p) p \, dp \geq \frac{1}{2}
\end{equation*}
as $f_P(p)$ is a valid pdf and $1-p \geq 1/2$ on the interval $[0,0.5]$ and $p \geq 1/2$ on $[0.5,1]$. Without loss of generality assume that $\Exp[P]\leq 1/2$. (The case that $\Exp[P]>1/2$ follows similarly due to the symmetry of $H_2(p)$.) Due to the aforementioned symmetry of $H_2$ around $1/2$ it follows that
\begin{equation*}
H_2\left(\int_0^{0.5} f_P(p)(1-p)\, dp + \int_{0.5}^1 f_P(p) p \, dp\right) \leq H_2\left(\int_0^1 f_P(p) p \,  dp\right)
\end{equation*}
if
\begin{equation}\label{eq:comparison_encoder_CSI_no_CSI}
1-\left(\int_0^{0.5} f_P(p)(1-p) \, dp + \int_{0.5}^1 f_P(p) p \, dp\right) \leq \int_0^1 f_P(p) p \, dp \enspace .
\end{equation}
We show that inequality~\eqref{eq:comparison_encoder_CSI_no_CSI} holds and analyze for which cases equality is reached.
\eqref{eq:comparison_encoder_CSI_no_CSI} is equivalent to
\begin{align}
1-\left(1-\int_{0.5}^1 f_P(p)\, dp\right) + \int_0^{0.5} f_P(p) p \, dp - \int_{0.5}^1 f_P(p) p \, dp &\leq \int_0^1 f_P(p) p \, dp \nonumber \\
\int_{0.5}^1 f_P(p) \, dp &\leq 2 \int_{0.5}^1 f_P(p) p \, dp \nonumber \\
\int_{0.5}^1 f_P(p) (1-2p) \, dp &\leq 0 \enspace ,
\end{align}
where the last line holds because the integrand can never be positive and equality only holds if $f_P(p)=0, \; \forall p \in (1/2,1]$ almost everywhere.

\subsection{Non-causal Channel State Information at the Encoder}

\begin{proposition}
	For the capacity of the VBSC with non-causal CSI at the encoder denoted as $C_{VBSC-E/nc}$ it holds that
	\begin{equation*}
		C_{VBSC-E} \leq C_{VBSC-E/nc} \leq C_{VBSC-ED} = C_{VBSC-D}
	\end{equation*}
\end{proposition}
\begin{proof}
	Both bounds are trivial. It is obvious that non-causal CSI at the encoder can only increase capacity compared to the case with causal CSI and clearly CSI at encoder and decoder can just increase the capacity compared to the case for non-causal CSI only at the encoder.
	\qed
\end{proof}

\section{Results for a Fixed Crossover Probability Distribution}\label{sec:example}
In this section, we take the model for a specific crossover probability distribution introduced in \cite{Mae13} (see Fig.~\ref{fig:pdf}) and compute the capacities for the CSI models discussed in Subsections~\ref{subsec:no_csi} to \ref{subsec:csi_causal_enc}.

To obtain the desired results, we use Theorems~\ref{th:noCSI_encoder_decoder}, \ref{th:capacity_encoder_decoder}, \ref{th:csi_decoder} and \ref{th:VBSC_cap_encoder} as well as Corollary~\ref{cor:encoder_decoder}. The results are presented in Table~\ref{tab:results}. We observe that CSI at encoder and decoder increases the capacity of the VBSC for the proposed crossover probability distribution by $0.179$ bits per channel use (or by 25\%) compared to the case without CSI at encoder and decoder. Table~\ref{tab:results} furthermore shows that causal CSI at the encoder increases the channel capacity by $0.0688$ bits per channel use again compared to the case without CSI.

\begin{table}[t]
	\caption{Numerical computation of the capacities for the pdf $f_P(p)$ from \cite{Mae13} which is depicted in Fig.~\ref{fig:pdf}.}
	\centering
	\begin{tabular}{|P{2.2cm}|P{2.2cm}|P{2.2cm}|P{2.2cm}|}
		\hline
		$C_{VBSC}$ & $C_{VBSC-E}$ & $C_{VBSC-D}$ & $C_{VBSC-ED}$ \\  \hline
		$0.6961$ bpcu & $0.7649$ bpcu & $0.8751$ bpcu & $0.8751$ bpcu\\ \hline
	\end{tabular}
	\label{tab:results}
\end{table}

\newpage
\section{Conclusion}\label{sec:conclusion}
In this work we have introduced the Varying Binary Symmetric Channel (VBSC) to describe the difference of PUF responses between the key enrollment and the key reproduction phase. We have derived capacity results depending on the available channel state information (CSI) at encoder and decoder. To exemplify our results we computed the channel capacities for the crossover probability model proposed in \cite{Mae13}. The results show that the capacity of the corresponding VBSC increases by $25\%$ if encoder and decoder have CSI compared to the case when both have no knowledge about the channel state. Furthermore, we argued that polar codes can be used to achieve capacity if CSI is available at encoder and decoder.

\bibliographystyle{splncs04}
\bibliography{main}

\end{document}